%% file: 2022-acc-multiagentPathOptimization.tex
\documentclass[letterpaper, 10 pt, conference]{ieeeconf}  
\usepackage{amsmath}
\DeclareMathOperator{\Pkchi}{P_k(\chi)}

\IEEEoverridecommandlockouts  % This command is only needed if you want to use the \thanks command

\input{preamble/common}
\newcommenter{zwang}{Firebrick1}
\newcommenter{sanderss}{BurntOrange}
\title{\LARGE \bf
Bearing-Based Formation Control with Optimal Motion Trajectory
}

\author{Zili Wang$^{}$, Sean B. Andersson$^{}$, and Roberto Tron$^{}$% <-this % stops a space
\thanks{Z.~Wang is with the Systems Engineering Division. S.B.~Andersson and R.~Tron are with the Systems Engineering Division and Mechanical Engineering Department at
        Boston University, Boston 02215,
        {\tt\small \{zw2445,sanderss,tron\}@bu.edu}.}%
\thanks{This work was supported in part by NSF through ECCS 1931600 and NRI 1734454 and through a grant from the Center for Information and Systems Engineering and the Boston University College of Engineering.}% <-this % stops a space
}
\begin{document}

\maketitle
\thispagestyle{empty}
\pagestyle{empty}

%%%%%%%%%%%%%%%%%%%%%%%%%%%%%%%%%%%%%%%%%%%%%%%%%%%%%%%%%%%%%%%%%%%%%%%%%%%%%%%%
\begin{abstract}
     Bearing-based distributed formation control is attractive because it can be implemented using vision-based measurements to achieve a desired formation. Gradient-descent-based controllers using bearing measurements have been shown to have many beneficial characteristics, such as global convergence, applicability to different graph topologies and workspaces of arbitrary dimension, and some flexibility in the choice of the cost. In practice, however, such controllers typically yield convoluted paths from their initial location to the final position in the formation. In this paper we propose a novel procedure to optimize gradient-descent-based bearing-based formation controllers to obtain shorter paths. Our approach is based on the parameterization of  the cost function and, by extension, of the controller. We form and solve a nonlinear optimization problem with the sum of path lengths of the agent trajectories as the objective and subject to the original equilibria and global convergence conditions for formation control. Our simulation shows that the parameters can be optimized from a very small number of training samples (1 to 7) to straighten the trajectory by around 16\% for a large number of random initial conditions for bearing-only formation. However, in the absence of any range information, the scale of the formation is not fixed and this optimization may lead to an undesired compression of the formation size. Including range measurements avoids this issue and leads to further trajectories straightening by 66\%.%With the additive range measurements, the trajectory can be further straightened by 60\%. 
\end{abstract}

\section{Introduction}
The goal of multi-agent formation control is to use distributed control to drive a number of agents to a desired geometric pattern. The advantages of such a system include the possibility of controlling a large network by a single operator, and the robustness of the system to failures of a single agent. As a results, formation control approaches have been widely applied to a variety of applications including surveillance, exploration, and transportation \cite{Anderson:RALC2008,Michael:AR2011,Petitti:ICRA2016}.

Compared to the extensively studied distance-based formation control \cite{Oh:Automatica2015}, bearing-based formation methods are appealing since relative bearing measurements are easy to obtain from an on-board camera and they are often more reliable than distance measurements, especially when the agents are relatively far away.

\textbf{Review of prior work.}
While there has been significant work in formation control, we limit ourselves to the progress in bearing-based formation control. The initial work of distributed bearing-based formation control traces back at least two decades  \cite{Bishop:CDC11,Bishop:CCA13,Bishop:JRNC2015}. However, the distance corresponding to each bearing measurement is required by the control law. The control method proposed in \cite{Franchi:IJRS2012} provided fast and straight trajectories using zero or one distance measurement, but it relied either on special graph topologies based on two leader agents, or on distributed estimators that virtually realize the measurements of such topologies. A more recent method was based on the gradient control law that does not depend on the distances, with the stability analysis relying on the state of the entire network evolving on a sphere \cite{Zhao:ITAC2016,Zhao:CSM2019}.

Following the above work, another gradient-based control law that requires only relative bearing measurements, but can also handle optional distance measurements was proposed in \cite{Tron:CDC2016} (which extended an earlier work on the visual homing method from \cite{Tron:ICRA15}). The control is based on the gradient of a Lyapunov function; global convergence is guaranteed by imposing constraints on this function. Advantages of this method include the fact that it can cover arbitrary numbers of agents, graph topologies and workspace dimensions. 
Several works have extended \cite{Tron:CDC2016} to more settings including second order dynamics \cite{Tron:CDC2018}, double integrator and unicycle dynamics \cite{Zhao:ICCA18, Zhao:TAC2019}, directed acyclic graphs and directed cycle graphs \cite{Karimian:ACC20}. All these works are based on a particular choice of the Lyapunov function defining the gradient-based controller, despite the fact that the stability conditions in the original work \cite{Tron:CDC2016} allow some flexibility in this regard. 

Other prior work has established a solid theoretical footing for the approach, establishing the concept of bearing rigidity theory \cite{Zhao:CSM2019} and developing bearing based control laws \cite{Bishop:JRNC2015,Zhao:ITAC2016,Tron:CDC2016,Zhao_bearingOnly:CSL2021}. However, a notable gap in the previous literature is that, while ensuring global convergence, there is no attempt to characterize, let alone optimize, the length and smoothness of the trajectories of the multi-agent system during their transitory phase. Typical trajectories of this controller are shown in the examples of Fig.~\ref{fig:5agents_1par_init},~\ref{fig:fuller_5agents_1par_init}.

A natural objective, then, is to optimize the path length of the agents, ideally reaching straight lines (i.e, the shortest possible path). Such a problem is not straightforward since each agent has very limited information about the global state of the network: it only observes its neighbors, moreover, it may observe only the directions, which is a nonlinear function of the state. 
As such, it is not trivial to find a controller that guides the agents along the shortest linear paths. In this paper, we focus on finding controllers that minimize the agents' path lenght for gradient-based bearing-based formation control with optional range measurements.

\begin{figure}[htbp]
 \centering
  \subfloat[]{\includegraphics[width=4.2cm]{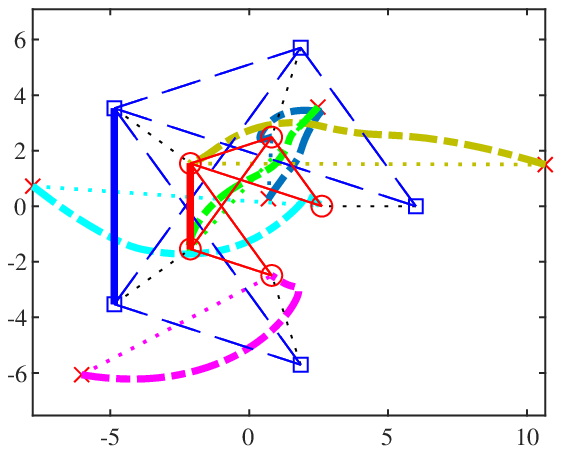}\label{fig:5agents_1par_init}}\quad
  \subfloat[]{\includegraphics[width=4.05cm]{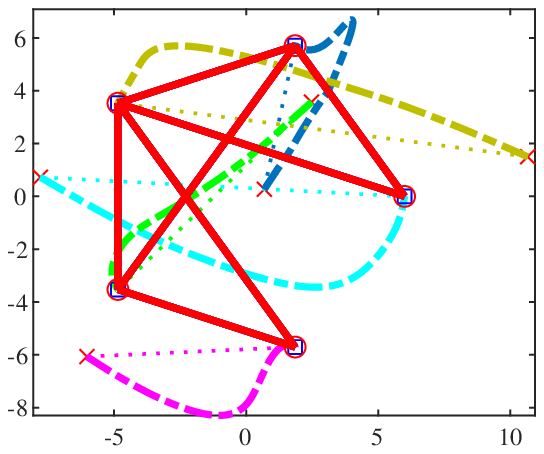}\label{fig:fuller_5agents_1par_init}}\quad
%   \subfloat[]{\includegraphics[width=4.2cm]{figures/homo/FullEr_2agents_1par_init.eps}\label{fig:fuller_2agents_1par_init}}

  \subfloat[]{\includegraphics[width=4.2cm]{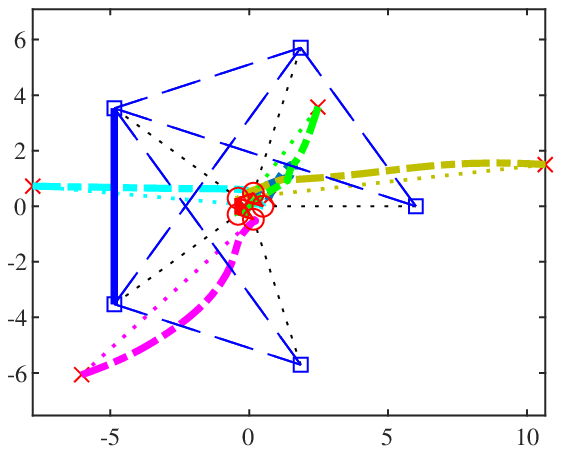}\label{fig:5agents_1par_opt}}\quad
  \subfloat[]{\includegraphics[width=4.05cm]{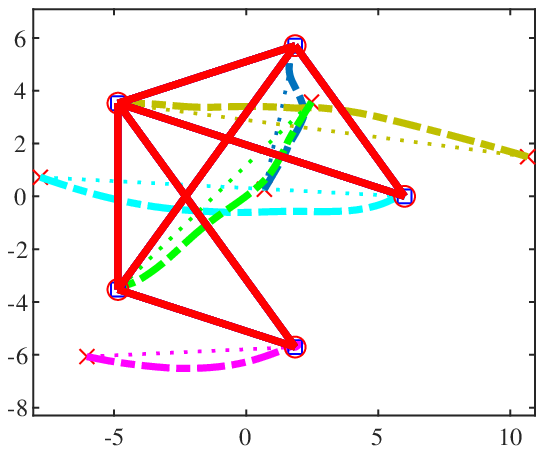}\label{fig:fuller_5agents_1par_opt}}\quad
%   \subfloat[]{\includegraphics[width=4.2cm]{figures/homo/FullEr_2agents_1par_opt.eps}\label{fig:fuller_2agents_1par_opt}}
  \caption{(a)(c): bearing-only 5-agent sytem where $E_d=\emptyset$, (b)(d): bearing-based 5-agent system where $E_b=E_d$. Agents' given desired formation \mysquare{blue} are moving via thick dashed lines from initial \mycross{red} to desired \mycircle{red} locations. First row is before the optimization and second row is after. } 
  \label{fig:5agents_1par}
\end{figure}

\textbf{Paper contributions.}
Our key contributions are:
\begin{itemize}
    \item We modify the convergence proof in \cite{Tron:CDC2016} to relax the conditions on the control cost for bearing-only formation and thereby expand the class of functions that can be used in the controller that guarantee global convergence.
    \item While maintaining global convergence, we formulate a constrained nonlinear optimization problem to find parameters for the Lyapunov function that leads to straighter trajectories. 
    \item We apply the optimization to bearing-based controllers and highlight an undesirable collapse effect that appears when using minimum path length as the objective. We further show that by including any number of range measurements in the controller, this effect can be avoided. %with different number of range measurements, and show that minimizing bearing-only formation trajectory path gives undesired side effects.
\end{itemize}
To evaluate the optimization results, we use two performance metrics: percentage of improvement in trajectory length relative to the original path, and percentage of improvement in trajectory length difference relative to an ideal straight-line path. Through simulations, we demonstrate that for these two metrics, our method improves the performance by approximately 8\% and 16\%, respectively,  in a simulated 5-agent bearing-only network. Including additional range measurements leads to an improvement of approximately 11\% and 66\% respectively. We also show that even with a small number of training initial conditions, the improvements generalize to (1) a far larger number of testing initial conditions, (2) a different number of agents, and (3) various target formation shape.
%%%%%%%%%%%%%%%%%%%%%%%%%%%%%%%%%%%%%%%%%%%%%%%%%%%%%%%%%%%%%%%%%%%%%%%%%%%%%%%%

\section{Preliminaries} 
We denote the dimension of the work-space with $n$. We assume that agents acquire bearing measurements in their own local reference frame, but that all the local reference frames are rotationally aligned (equivalently, the agents know their rotation with respect to a common global frame \cite{Leonardos:CDC2019}), and all the quantities discussed are expressed in a common global inertial reference frame (e.g., by assuming the availability of a compass for each agent). The operator $\stack(v_1,\dots,v_m) = \stack(\{v_i\}_{i=1}^m)$ returns the vector obtained by vertically stacking the arguments.
\subsection{Graph Theory}
The interaction topology of a multi-agent system is modeled as a directed graph $\mathcal{G}=(\mathcal{V},\mathcal{E})$, where $\mathcal{V}$ denotes the set of nodes, $\mathcal{E}\subseteq \mathcal{V}\times \mathcal{V}$ denotes the edges of ordered pairs of the nodes, and the set of neighbors of $i\in \mathcal{V}$ is $\mathcal{N}_i:=\{j\in \mathcal{V}: (i,j)\in \mathcal{E}\}$. A graph is undirected if for every $(i,j)\in \mathcal{E}$, $(j,i)$ is also in $\mathcal{E}$. We assume all graphs are undirected (as it will be seen from the definitions in the following section, measurements for the edge $j\to i$ can be easily obtained from those for edge $i\to j$ via communication).
\subsection{Formations and Measurements} \label{section2: formation definitions}
We represent the set of $N$ agents as $\cV=\{1,\dots,N\}$ and the corresponding location of each agent as $\{x_i\}_{i\in \cV}$. We define the \emph{range} between nodes $i,j\in \cV$ as
\begin{equation}
    d_{ij}(x_i,x_j)=\norm{x_j-x_i},\label{eq:distance}
\end{equation}
where $\|\cdot\|$ denotes the Euclidean norm. The \emph{bearing direction} is defined as
\begin{equation}
    \beta_{ij}(x_i,x_j)=d_{ij}^{-1}(x_j-x_i).\label{eq:bearing}
\end{equation}

A bearing-based formation is defined as $\cF = (\cG,\vct{x})$ where
\begin{itemize}
    \item $\vct{x}=\stack(\{x_i\}_{i\in \cV})$ is the \emph{configuration} of the formation that gives the location of each agent in $\mathbb{R}^n$.
    \item $\cG = (\cV,\cE_b,\cE_d)$ is a graph in which $\cE_b$ contains the set of pairs $(i,j)$ where agent $i$ can observe $\beta_{ij}$ and $\cE_d$ contains the set of pairs that observe $d_{ij}$. We assume that $\cE_d \subset \cE_b$.
\end{itemize}
$\cF$ is a bearing-only formation if $\cE_d=\emptyset$.
 The complete set of bearings and ranges of a formation is denoted by the vector  $\boldsymbol\beta=\stack(\{\beta_{ij}\}_{(i,j)\in \cE_b})$, $\mathbf{d}=\stack(\{d_{ij}\}_{(i,j)\in \cE_d})$. Two formations $(\mathcal{G},\vct{x})$ and $(\mathcal{G},\vct{x'})$, are said to be
\begin{itemize}
    \item \emph{equivalent} if they yield the same measurements, $\boldsymbol{\beta}\!=\!\boldsymbol{\beta'}$, $\mathbf{d} = \mathbf{d'}$.
    \item \emph{congruent} if they have the same shape and scale, $x_i'=x_i+t$ with $t\in\mathbb{R}^n$ and for all $i \in \cV$.
    \item \emph{similar} if they have the same shape, $x_i'=\gamma x_i+t, \gamma>0$, with $t\in\mathbb{R}^n$ and for all $i \in \cV$.
\end{itemize}

A formation $(\cG,\vct{x})$ is rigid if every configuration equivalent to it is also similar (for bearing-only) or congruent (for bearing+range) to it. In this paper, we assume that all formations are rigid. 

\subsection{Gradient-based Formation Control}\label{section:gradient-based control}
We take a simple kinematic motion model,
\begin{equation} \label{eq:motion}
  \vct{\dot{x}} = \vct{u},
\end{equation}
and follow the derivation of the distributed gradient-based controller from \cite{Tron:CDC2016}. The control is defined as
\begin{equation} \label{eq:control}
  \vct{u} = -\frac{\partial \varphi(\vct{x})}{\partial x},
\end{equation}
where the cost function $\varphi$ for a bearing-only formation is
\begin{align} 
  &\varphi(\vct{x}) = \varphi_b(\vct{x}) = \sum\limits_{(i,j) \in \cE_b} \varphi_{ij}^b (x_i, x_j), \label{eq:control cost}\\
  &\varphi_{ij}^b(x_i,x_j)=d_{ij}f_b(c_{ij}), \label{eq:edge cost}\\
  &c_{ij}(x_i,x_j) = \beta_{g,ij}^T \beta_{ij}
    = \cos(\angle(\beta_{g,ij}, \beta_{ij})),
\end{align}
where $\angle(\cdot,\cdot)$ denotes the angle between two vectors, $c_{ij}$ is referred to as a \emph{bearing similarity} between the current bearing $\beta_{ij}$ and desired bearing $\beta_{g,ij}$, and $f_b(c_{ij})$ as a \emph{bearing reshaping function}. In Sec.~\ref{section:fb conditions}, we derive a set of conditions on the bearing reshaping function to ensure global stability of the control law to the desired formation that are less restrictive than those found in \cite{Tron:CDC2016}.
%The conditions on $f_b(c_{ij})$ for global convergence of the formation controller will be obtained in Sec.~ \ref{section:fb conditions}. 
The cost function \eqref{eq:control cost} is a summation over the edges $\cE_b$, with each term being a monotonic function of the similarity between  the current and desired bearing. 
The control law \eqref{eq:control} for any specific agent $i \in \cV$ can be written as
\begin{equation}
    u_i = -\sum\limits_{j:(i,j) \in \cE_b} g^b_{ij} (x_i, x_j), \label{eq:agent cost}
\end{equation}
where $g^b_{ij}$  denotes the gradient of \eqref{eq:edge cost} and is given by the following (a detailed derivation can be found in \cite{Tron:ICRA15}):
\begin{equation}
    g_{ij}^b = -f_b(c_{ij})\beta_{ij} - f'_b(c_{ij})(I_n-\beta_{ij}\beta_{ij}^T)\beta_{g,ij} \label{eq:bearing cost gradient},
\end{equation}
where $I_n\in\mathbb{R}^{n\times n}$ denotes an identity matrix, and $f'_b(c_{ij})$ is the derivative of $f_b$ evaluated at $c_{ij}$. Note that although the unmeasured distance for each edge appears in the cost function, it does not appear in the bearing-only control law.

%The scaling problem of the above controller can be eliminated by adding range measurements, the cost function \eqref{eq:control cost} then becomes
The bearing-based controller adds the range measurements to the cost function, yielding
\begin{align} 
  &\varphi(\vct{x}) = \beta_b \sum\limits_{(i,j) \in \cE_b} \varphi_{ij}^b (x_i, x_j) + \beta_d \sum\limits_{(i,j) \in \cE_d} \varphi_{ij}^d (x_i, x_j), \label{eq:control cost with range}\\
  &\varphi_{ij}^d(x_i,x_j)=f_d(q_{ij}), \label{eq:edge cost with range}\\
  &q_{ij}(x_i,x_j) =d_{ij}c_{ij}-d_{g,ij},
\end{align}
where $q_{ij}$ is the \emph{range similarity} between the current and desired range. The conditions on the \emph{range reshaping function} $f_d(q_{ij})$ to ensure global convergence of the agents to the desired formation using the cost \eqref{eq:control cost with range} are: (1) $f_d(q)> 0$, $f_d(0)=0$, (2) $\text{sign}(f'_d(q))=\text{sign}(q)$, (3) $f''_d(0)>0$ (see, e.g. \cite{Tron:CDC2016}).

\begin{remark}\label{remark:discontinuity}
    At the point where the controller is undefined (i.e. when $x_i=x_j$), the edge cost \eqref{eq:edge cost} becomes $0$ by continuity and one of the subgradient is $\vct{0}_{2n}$.
\end{remark}

\section{Conditions for Bearing-only Global Convergence}\label{section:fb conditions}
We focus primarily on the bearing-only approach and provide the first contribution of the paper, namely a less restrictive set of constraints on the bearing reshaping function relative to those of \cite{Tron:CDC2016} that ensure global convergence. First, we want to ensure that the gradient-based bearing-only controller converges to equilibria where the configuration is similar (in the sense defined in Sec.~\ref{section2: formation definitions}) to the desired one. This implies that every minimum of $\varphi(\cdot)$ should be a global minimum (i.e., there are no spurious local minima). Following the reasoning of \cite{Tron:CDC2016}, such a condition can be achieved by requiring certain properties on the individual reshaping function $f_b(\cdot)$; an advantage of this approach is that then convergence conditions will be independent of the graph topology. Below, we derive a set of conditions on $f_b(\cdot)$ that is significantly less restrictive than the one provided in \cite{Tron:CDC2016}. Specifically, we simply require that $f_b(\cdot)$ is non-negative and monotonically decreasing:
\begin{assumption}\label{assumption:constraints}
The bearing reshaping function $f_b(c_{ij}): [-1,1] \rightarrow \mathbb{R}$ has the following properties 
\begin{align} 
    &f_b(1) = 0, \label{eq:bearing constraints_original1}\\
    &f'_b(c_{ij})= 
        \begin{cases}
           \leq 0,& \text{and finite for } c_{ij} = 1,\\
            < 0,              & \text{otherwise}.
        \end{cases} \label{eq:bearing constraints_original2}
\end{align}
\end{assumption}
Assumption \ref{assumption:constraints} is similar to that of \cite{Tron:CDC2016} but we have removed the condition $f_b(c_{ij}) + (1-c_{ij})f'_b (c_{ij}) \leq 0$, which significantly increases the set of functions that can be chosen.

%From here, we first show that $\varphi(\cdot)$ can be used as a formation controller, by proving that ith h h has a global minimizer at desired bearings.
We now need to establish that the set of global minima of $\varphi^b(\cdot)$ corresponds to the goal of formation control:
% It was proved that for bearing-only observations, a similar configuration to the desired one can be achieved. However, an optimal trajectory is not considered. In the next sections, we will introduce our methods of optimize the gradient-based bearing-only formation controller. 

%%%%%%%LEMMA 1
\begin{lemma}\label{lemma1}
    The cost function $\varphi^b(\cdot)$ is non-negative everywhere and has a global minimizer at configurations $\vct{x}$ that are similar to the desired configuration $\vct{x_g}$.
\end{lemma}

The proof of this lemma can be found in \cite[Lemma 1]{Tron:CDC2016}. We then show the key property for proving global convergence:
%%%%PROPOSITION 1
\begin{proposition} \label{prop:1}
        The cost function $\varphi^b(\cdot)$ has only global minimizers at configurations similar to $\vct{x_g}$, and there are no other critical points.
\end{proposition}

The claim is the same as \cite[Prop.~1]{Tron:CDC2016}, but we now establish it under the conditions of Assumption~\ref{assumption:constraints}.
%applies with less restrictive conditions of this paper. 
Since $\varphi(\cdot)$ is generally non-convex, in order to prove Prop. \ref{prop:1}, we first provide a lemma that evaluates the cost function on a parametric line starting from a common point $x_0$ and moving in arbitrary directions, showing that is increasing except in the direction of the desired bearing $\beta_{g,ij}$.
%%%%%%%LEMMA 2
In the following statements, we use the notation $\tilde{\cdot}$ to indicate a function evaluated along a curve $\tilde{x}$.
\begin{lemma}\label{lemma2}
    Define the line $(\Tilde{x}_i(t),\Tilde{x}_j(t))=(x_{0}+t v_i,x_{0}+tv_j)$, where $x_0$ is an arbitrary point, and $v_i,v_j\in \mathbb{R}^n$ are arbitrary directions. The derivative of the function 
  \begin{equation}
      \Tilde{\varphi}_{ij}^b(t)=\varphi_{ij}^b(\Tilde{x}_i(t),\Tilde{x}_j(t))
  \end{equation}
  satisfies the following
  \begin{equation}
    \begin{aligned}
    \dot{\Tilde{\varphi}}_{ij}^b
        \begin{cases}
           \equiv 0,& \text{if } c_{ij}(x_i,x_j) = 1,\\
            > 0,    & \text{otherwise}.
        \end{cases}
    \end{aligned}
  \end{equation}
\end{lemma}

The proof of this lemma is given in Appendix \ref{Appendix1}. The concept is similar to \cite[Lemma 1]{Tron:CDC2016}, except that here we have $\tilde{x}_i$ and $\tilde{x}_j$ start from the same location rather than requiring an offset between them. With this lemma, we can prove Prop.~\ref{prop:1}:

\begin{proof}
    Let $\vct{x_g}$ be consistent with $\mathcal{F}$ and bearings $\vct{\beta_g}$ (i.e. $\vct{\beta}(\vct{x_g})=\vct{\beta_g}$). We can define a parametric line $\Tilde{\vct{x}}(t) = \vct{x_g}+t(\vct{x_0}-\vct{x_g})$, where $\vct{x_0}=\Tilde{\vct{x}}(1)$, $\vct{x_0}\neq\vct{x_g}$, is an arbitrary configuration. By linearity, we have
    \begin{equation}
        \dert \varphi^b(\Tilde{\vct{x}}(t))\Bigg|_{t=1} = \sum_{(i,j)\in \cE_b} \dert\varphi_{ij}^b(\Tilde{\vct{x}}(t))\Bigg|_{t=1}.
    \end{equation}
    Lemma \ref{lemma2} shows that each term on the right hand side is non-negative, and zero if and only if $\beta_{ij}=\beta_{g,ij}$, that is, at configurations similar to $\vct{x_g}$. We deduce that the Lie derivative of $\varphi^b$ in the direction $\vct{v}=\vct{x}_0-\vct{x}_g$ at $\vct{x}_0$ , $\frac{\partial}{\partial x}\varphi^b(\vct{x_0}) \vct{v}=\dert \varphi^b(\Tilde{\vct{x}}(t))\big|_{t=1}$, is strictly positive. Hence $\frac{\partial}{\partial x}\varphi^b(\vct{x_0})\neq 0$ and the configuration $\vct{x_0}$ is not a critical point unless $\vct{x_0}$ is similar to $\vct{x_g}$.
\end{proof}

From the above statements, we provide our result of global convergence on the proposed controller. 

%%%%%THEOREM 1
\begin{theorem}
    Every trajectory of the closed-loop system
    \begin{equation}
        \vct{\dot{x}}(t) = -\frac{\partial}{\partial x} \varphi^b(\vct{x}(t))
    \end{equation}
    asymptotically converges to a configuration $\vct{x}$ similar to the desired configuration $\vct{x_g}$.
\end{theorem}
\begin{proof}
     The claim is a restatement of \cite[Theorem 1]{Tron:CDC2016} but using our less restrictive conditions on $f_b(c_{ij})$ based on Lemma~\ref{lemma2} instead of \cite[Lemma 2]{Tron:CDC2016} and \cite[Lemma 3.4]{Tron:ICRA15}. Due to space limitations, we do not provide a proof here.% The essence of the proof is to first show that the trajectories are not diverging toward infinity, and then show convergence to the set of global minima.
\end{proof}

\begin{remark}
    We note that points where $x_i=x_j$ can lead to a zero in the corresponding component of the Lie derivative of $\varphi^b$; so long as other edges in the controller have nonzero derivative this does not pose an issue.
    %The consensus point where $x_i=x_j$ is also considered as desired configuration. In practice, the system will not converge to that point. However, it can converge under the optimization. We will discuss it mainly in Section \ref{Section:results}.
\end{remark} 
%\sanderss{Huh? Do you mean the point where $x_i=x_j$ for all $i,j$ is similar to every configuration? But that configuration is not even defined, right? In your results, do you really see things collapsing to exactly everything on top of one another, or to closer and closer together? Regardless, I am afraid that this remark will backfire and recommend just removing it.}
%\zwang{bearing+range, range is convex, so itself is a local minimum. However, cannot prove it for bearing+range case with relaxed $f_b(\cdot)$ condition yet.}

%%%%%%%%%%%%%%%%%%%%
%%%%%%%%%%%%%%%%%%%%
\section{Nonlinear Optimization of Formation Control}
In this section, we give the details of our second contribution. We formulate a constrained nonlinear optimization problem using a combination of the path length of the agent trajectories and the control cost at terminal time. This choice seeks to minimize the trajectory length while enforcing a common convergence time. By describing the bearing reshaping functions in a parameterized form (with a similar parameterization for the range reshaping functions when they are included), we obtain a parameter optimization problem, subject to the constraints in Assumption~\ref{assumption:constraints}.
%We formulate a constrained nonlinear optimization problem based on the functions $f_b(c_{ij},\alpha_b)$ and optional extension on functions $f_d(q_{ij},\alpha_d)$, which are piecewise quadratic polynomials with parameters $\alpha$ ($\alpha=\alpha_b$ if bearing-only, $\alpha=\left(\alpha_b,\alpha_d\right)$ if bearing+distance). We choose the objective function as a combination of path length and control cost at the terminal time. Such a choice minimizes the trajectory length while enforcing a common convergence time. We define the optimization problem as minimizing the selected objective over the parameters $\alpha$, subject to the constraints in Assumption~\ref{assumption:constraints}. 
Then we use the sensitivity function to find the derivative of the objective function, which can be used in a nonlinear optimization solver.
\subsection{Problem Definition}
The objective function for our problem is
%We define the objective function for our problem as a weighted sum of a path length term and a terminal formation error term under an initial condition $x_0$
\begin{equation}\label{eqn:objective function}
    L(\vct{\alpha},\vct{x}_0) = \left. \sum_{i\in \cV}\int^T_0 ||\dot{x}_i (\vct{x}(t),\vct{\alpha)}|| \de t+\omega\varphi(\vct{x}(T)) \right|_{\vct{x}(0)=\vct{x}_0},
\end{equation}
where $\vct{x}_0$ is the initial condition on the agents, $\vct{\alpha} = \alpha_b$ are the parameters defining the reshaping functions $f_b(c_{ij},\alpha_b)$% and $f_d(q_{ij},\alpha_d)$
, $0$ and $T$ are the starting and terminal time, and $\omega$ is a weight to balance the two terms. (If the bearing-based control is being used, then $\alpha$ will also include the parameters $\alpha_d$ for the range-reshaping functions.)

In general, the optimal $\alpha$ is dependent on the initial conditions of the system. To achieve good performance across a range of initial conditions, we optimize over a finite number of randomly selected initial conditions, defined as the set $\cX_0$.
%we generalize $f_b(\cdot)$ by using a batch of systems with different random initial positions in a set $X_0$ to update on the optimization problem. 
Our problem is then
\begin{equation} \label{eqn: otpimization problem}
\begin{aligned}
    & \underset{\vct{\alpha}}{\min}
    & & \sum_{\vct{x}_0\in \cX_0}L(\vct{\alpha},\vct{x}_0) \\
    & \text{subj. to}
    & & (\ref{eq:motion}), (\ref{eq:control}), (\ref{eq:bearing constraints_original1}), (\ref{eq:bearing constraints_original2}),
\end{aligned}
\end{equation}
where \eqref{eq:motion} and \eqref{eq:control} are the motion model and control law by using parametrized $f_b(c_{ij},\alpha)$, \eqref{eq:bearing constraints_original1} and \eqref{eq:bearing constraints_original2} are the required conditions on $f_b(\cdot)$. The conditions on $f_d(\cdot)$ are required if we also parameterize the range terms.

\subsection{Function Interpolators}\label{section:interpolator}
% We want to find a way to parametrize the reshaping function which can be plugged into the optimization problem. 
There are many approaches to parametrize a function $f(\chi,\alpha_b)$ defined on $[-1,1]$, given a set of control points located on a regular grid $1=\chi_1>\chi_2>\dots>\chi_K=-1$ with grid size $h = \chi_{k+1} - \chi_k$ ($k\in\{ 1,2,\dots,K-1\}$). We consider piecewise second order polynomial interpolators here for two reasons: (1) they remain numerically stable with respect to the number of grid points, (2) while the piecewise linear interpolation requires the smallest number of parameters and the least computation effort, in general it cannot guarantee a continuous first derivative $f'(\cdot)$.  %For this reason, it is prefereable to use a higher order piecewise polynomial interpolation scheme.  
A general quadratic polynomial can ensure continuous differentiability on the interval, with linear constraints \eqref{eq:bearing constraints_original1} and \eqref{eq:bearing constraints_original2}.

Given a quadratic function $Q_k(\cdot)$ defined on the interval $[\chi_k, \chi_{k+1}]$ and passing through the control points,
\begin{align}
     Q_k(\chi) &= a_k^0 + a_k^1 (\chi - \chi_k) + a_k^2 (\chi - \chi_k)^2, \\
    Q_k(\chi_k) &= a_k^0 = f(\chi_k,\alpha_b),
\end{align}
continuity of the first and second derivatives results in the following constraints on the vector of coefficients $\vct{a}=\stack(\{a^r_k)_{k\in\{1,\ldots,K\}}^{r\in\{0,1,2\}})$:
\begin{equation}\label{eq:value constraint}
    a_{k+1}^0 = Q_{k+1}(\chi_{k+1}) = Q_k(\chi_{k+1})= a_k^0 + a_k^1 h + a_k^2 h^2,
\end{equation}
\begin{equation}\label{eq:derivative constraint}
a_{k+1}^1 = Q'_{k+1}(\chi_{k+1}) = Q'_k(\chi_{k+1}) = a_k^1 + 2 a_k^2 h,
\end{equation}
for $k\in\{ 1,2,\dots,K-1\}$~\cite{rburden11:numerical}. The coefficients can thereby be written as $\vct{a}=F\alpha_b$, where $\alpha_b\in\real{K+1}$ is a minimal set of parameters, and $F$ is a matrix such that $\vct{a}$ satisfies \eqref{eq:value constraint} and \eqref{eq:derivative constraint} for any $\alpha_b$.

The constraint $f_b(\chi_1,\alpha_b)=0$ from \eqref{eq:bearing constraints_original1} implies $a_1^0=0$. The reshaping function can then be represented as
\begin{align}
%\begin{split}
    f_b(\chi,\alpha_b) &= \sum^{k-1}_{i=1}(a_i^1 h + a_i^2 h^2)  + a_k^1 (\chi-\chi_k) + a_k^2 (\chi-\chi_k)^2 \nonumber\\
    &\triangleq \Pkchi^T F \alpha_b \label{eq:interpolator}
%\end{split}
\end{align}
on the interval $[\chi_k, \chi_{k+1}]$,
where $\Pkchi$ is an appropriate vector of polynomials of  $h$ and $\chi$ up to order two. 

Given the above parametrization of the reshaping function $f_b$, it is possible to show that the constraints (\ref{eq:bearing constraints_original1}) and (\ref{eq:bearing constraints_original2}) are satisfied  for all $\chi\in[-1,1]$ if and only if the following conditions are satisfied on the grid points $\{\chi_k\}$
\begin{align}
    %f
     f_b(\chi_1, \alpha_b) &= 0,\\
    %df
    f'_b(\chi_1,\alpha_b) &\leq 0, \\
    f'_b(\chi_k,\alpha_b) &< 0\; \forall k\in\{2,\ldots,K-1\}.
    %&f_b(\alpha,\chi_k) + (1-\chi_k)f'_b(\alpha,\chi_k) \leq 0, k\geq 1 (\chi \in [-1,1])
\end{align}
Noting that these constraints are all linear in $\alpha_b$.
%We see that enforcing the constraints on the end points of each segment guarantees the function is non-negative and monotonically decreasing on the entire range $c_{ij}\in[-1,1]$. 

%For the cubic spline interpolator, in addition to the above constraints, one more nonlinear constraint to guarantee monotonic cubic polynomial is required.
%As for the last constraint, the derivative of $f_b$  is a constant on each segment, so the constraint itself is linear, enforcing it on the end points can also guarantee the satisfaction. However, enforcing all the constraints may not give straight trajectories. Therefore, the constraints are relaxed to some extent.
\begin{remark}\label{remark:cost weight}
The objective function \eqref{eqn:objective function} depends on the choice of the weighting parameter $\omega$. The optimization result, however, is insensitive to the specific value chosen. A detailed explanation is provided in Appendix \ref{Appendix2}. Here we use $\omega=1000$ in our simulation.
\end{remark}

\subsection{Derivative of the Objective Function}\label{section:objective derivation}
Since the objective function \eqref{eqn:objective function} is nonlinear, we propose to use an off-the-shelf Sequential Quadratic Programming (SQP) solver, which, in our tests, has shown better convergence times than alternatives. At a high level, this method models the problem at the current approximate solution by a quadratic programming subproblem; then the solution of the subproblem is used to construct a better approximate solution. At each iteration of this process, the Hessian of the associated Lagrangian function is approximated using gradient information and quasi-Newton updates \cite{boggs1995sequential}. In this Section, we compute the analytical derivative of the objective function, which will be used by the solver instead of relying on the numerical approximations.

The objective function \eqref{eqn:objective function} is comprised of the total travelled distance and terminal cost. From the chain rule, the gradient of the travelled distance function with respect to $\vct{\alpha}$ is

 \begin{equation}\label{eq:derivative_obj1}
     \begin{split}
         \frac{\partial~}{\partial \alpha}\int_0^T||\dot{x}_i(\mathbf{x}(t), \mathbf{\alpha})|| dt &= \int_0^T \frac{\partial ||\dot{x}_i||}{\partial x_i}^T \frac{\partial\dot{x}_i}{\partial \alpha} ~dt \\
    &=\int_0^T \frac{ \dot{x}_i^T}{||\dot{x}_i||} \frac{\partial\dot{x}_i}{\partial \alpha} ~dt.
     \end{split}
 \end{equation}

For a bearing-only formation, given the single integrator dynamics \eqref{eq:motion} and control law \eqref{eq:agent cost}, and exchanging the order of the time and partial derivatives, then the last term under the integral is given by%derivative of $\dot{x}_i(\cdot)$ :
\begin{equation}\label{eq:DAlpha}
%    &x_i(t,\mathbf{\alpha}) = x_{i,0} - \sum_{j:(i,j)\in E_b} \int_0^T g_{ij}(x_i,\mathbf{\alpha}), \\
    %%
    \dert \frac{\partial x_i(t,\mathbf{\alpha})}{\partial \alpha} = -\;\smashoperator{\sum_{j:(i,j)\in \cE_b}}\; \big(\left.\frac{\partial g^b_{ij}}{\partial x}\right|_{x=x(t,\mathbf{\alpha})}\hspace{-0.25cm}\frac{\partial x_i}{\partial \alpha}+\left.
    \frac{\partial g^b_{ij}}{\partial \alpha}\right|_{x=x(t,\mathbf{\alpha})}\big).
\end{equation}

The equation above defines an ODE in terms of the \emph{sensitivity function} \cite{khalil2002nonlinear} $S_i$, which is defined as:
\begin{align}\label{eq:sensitivity}
    \begin{split}
    &S_i(t) = \frac{\partial x_i(t,\mathbf{\alpha})}{\partial \alpha_b}.\\
    \end{split}
\end{align}

We can then rewrite  (\ref{eq:DAlpha}) as
\begin{equation}\label{eq:sensitivity derivative}
    \dot{S_i}
    =-\;\smashoperator{\sum_{j:(i,j)\in \cE_b}}\; \big(\left.\frac{\partial g^b_{ij}}{\partial x}\right|_{x=x(t,\mathbf{\alpha})}S_i +\left. \frac{\partial g^b_{ij}}{\partial \alpha}\right|_{x=x(t,\mathbf{\alpha})}\big),
\end{equation}
where, using~\eqref{eq:agent cost},
\begin{align}
    \frac{\partial g^b_{ij}}{\partial x} =& -d_{ij}^{-1}(f''_b(c_{ij},\mathbf{\alpha})(c_{ij} \beta_{ij} - \beta_{g,{ij}})\beta_{g,{ij}}^T \nonumber\\
    &+(f'_b(c_{ij},\mathbf{\alpha})c_{ij} - f_b(c_{ij},\mathbf{\alpha})) I_n) P_{\beta_{ij}}, \label{eq:dgdx}\\
    \frac{\partial g^b_{ij}}{\partial \alpha} =& -\frac{\partial f_b(c_{ij},\alpha)}{\partial \alpha}\beta_{ij} - \frac{\partial f'_b(c_{ij},\alpha)}{\partial \alpha} (I_n-\beta_{ij}\beta_{ij}^T)\beta_{g,ij}.
\end{align} 

In practice, \eqref{eq:motion} and \eqref{eq:sensitivity derivative} are solved simultaneously using a single ODE solver. Note that the dynamics for the sensitivity function defined above involve $f''_b(\cdot)$. This term is not well-defined for the piecewise-quadratic functions used here. To overcome this issue, we define $f''_b(\cdot)$ to be right continuous at each control point, and handle it as in Remark \ref{remark:discontinuity}. 
 
For the terminal cost term, from the control cost \eqref{eq:control cost} and the sensitivity function \eqref{eq:sensitivity}, the derivative of the terminal control cost in the objective function is
 \begin{equation}\label{eq:derivative2}
 \begin{split}
     \frac{\partial \varphi(\vct{x}(T))}{\partial \alpha} &= \frac{\partial \varphi_b}{\partial \alpha}+\frac{\partial \varphi_b}{\partial x}\frac{\partial x}{\partial \alpha} \\
    &=\sum_{(i,j)\in \cE_b}d_{ij}\frac{\partial f_b(c_{ij},\alpha)}{\partial \alpha}+\frac{\partial \varphi_b}{\partial x} S.
\end{split}
\end{equation}

\begin{remark}\label{remark:}
    When two agents meet, \eqref{eq:dgdx} becomes infinity. In practice, we split \eqref{eq:derivative_obj1} at the discontinuity point (i.e. when $x_i=x_j$) to get around the problem of exchange the order in multiple derivatives. In order to solve this issue, we approximate the derivative by stopping one step earlier before the meeting condition, propagating with Euler step to escape that region. See Appendix \ref{Appendix3} for a more detailed discussion.%A possibly more exact solution is proposed in Appendix \ref{Appendix2}.
\end{remark}

\subsection{Additional Range Terms}
Since a bearing-only controller is scale-invariant, the optimization process for the bearing reshaping functions may lead to a solution where the agents become arbitrarily close together. In order to avoid this undesirable ``collapse'', we can include
%In order to remove the collapse effect of bearing-only formation, we need 
at least one range measurement. The process for including the range terms in the optimization is similar, except that the representation of $a_1^0$ in \eqref{eq:interpolator} is represented by elements in $\alpha_d$ and from partial conditions in \ref{section:gradient-based control}. The corresponding parameterized constraints of $f_d$ are
\begin{align}
  %f
   f_d(0, \alpha_d) &= 0,\\
  %df
  f'_d(0, \alpha_d) &= 0,\\
  f'_d(\chi_k,\alpha_d) &< 0, \forall \chi_k<0,\\
  f'_d(\chi_k,\alpha_d) &> 0, \forall \chi_k>0,\\
  a^1_1 \geq 0&, ~ a^2_1>0,\\
  a^1_K \leq 0&, ~ a^2_K>0.
  %&f_d(\alpha,\chi_k) + (1-\chi_k)f'_b(\alpha,\chi_k) \leq 0, k\geq 1 (\chi \in [-1,1])
\end{align}
These constraints ensure a parabolic-like shape of the distance reshaping function, with the minimum at the origin. The derivative of the objective function follows analogously to Sec.~\ref{section:objective derivation} and is omitted here for space reasons.

Note that we use the relaxed constraints on the bearing reshaping function in this setting as well; proving global convergence under these constraints when the range is included is a topic of ongoing work. 

%The constraints on bearings are our simplied version, with convergence proof remained as ongoing work. Due to space reasons, we omit Section \ref{section:objective derivation} in this case, as it can be derived in an analogous way, which simply results in additional terms for the derivative. 

% \subsection{Heterogeneous Reshaping Function}\label{section:heterogeneous fb}
% Up until this point, we have assumed the agents share a common parameter set $\alpha$. Improved performance may be achieved if we allow optimization over each of the agents individually. We therefore propose the heterogeneous case, where each agent gets its own parameter set $\alpha_i$ and corresponding reshaping function $f_b(c_{ij},\alpha_i)$. Accordingly, we will have one gradient for each $\alpha_i$.

\section{Simulations}\label{Section:results}
In this section, we validate our proposed optimization method using a 5-agent network in two dimensions. Through our simulations, we aim to address four questions: 
\begin{enumerate}
    \item how much can we improve the trajectories according to the proposed performance metrics?
    \item how does the improvement depend on the number of initial conditions in the set $\cX_0$ for the optimization?
    \item What can we conclude by comparing the bearing-distance formation to bearing-only formation?
    \item How do the results found for specific formations generalize over other formations?
    \item How does our controller performance compare with other methods?
\end{enumerate}
In the following, the initial position of the system is randomized and the desired formation is an equilateral polygon. %with vertices on a circle. 

\subsection{Training} \label{training}
Throughout, we initialize the optimization by discretizing the initial bearing shape function $f_b(c_{ij})=\tfrac{{\rm arc}\cos^2(c_{ij})}{2}$ and optional range reshaping function $f_d(q_{ij})=\tfrac{q_{ij}^2}{2}$ using seven control points and interpolate them with quadratic interpolation to determine the initial parameter $\alpha$.
We ran simulations using 1, 3, 5, or 7 different initial conditions in $\cX_0$ for the initial positions of the five agents. There are four cases: (1) without range terms ($\cE_d=\emptyset$, labeled NoEd), (2) with one range term ($|\cE_d|=1$, labeled OneEd), (3) with three range terms ($|\cE_d|<|\cE_b|$, SomeEd), and (4) with seven range terms ($|\cE_d|=|\cE_b|$, FullEd). Performance of the optimization clearly depends on the choice of which edges to include; we show selections that yielded the best results. Fig.~\ref{fig:fbPlot} shows the optimized bearing reshaping functions while Fig.~\ref{fig:fdPlot} shows the optimized range reshaping function. From Fig.~\ref{fig:fbPlot}, we see that the optimized bearing reshaping functions are relatively flat when from the desired bearing ($c_{ij}$ far from 1) and fall off to zero as that desired bearing is reached. The optimized range reshaping functions retain their basic quadratic shape, becoming shallower with fewer range measurements.

%We infer that an optimal reshaping function is flat when $c_{ij}$ is far from 1 (that is, the current bearing is far away from the desired bearing) and transitions to a steep quadratic shape when $c_{ij}$ is close to 1 (as the desired bearing is reached). 
\begin{figure}[htbp]
  \centering
  \includegraphics{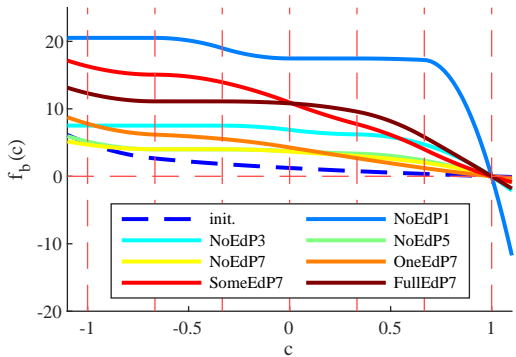}
  \caption{Bearing reshaping function before (dashed blue) and after (solid) optimization. \textasteriskcentered{}EdP\textasteriskcentered{} indicates the number of range edges and initial conditions. Vertical dashed lines show the control points for the interpolator.}
  \label{fig:fbPlot}
\end{figure}

\begin{figure}[htbp]
    \centering
    \includegraphics{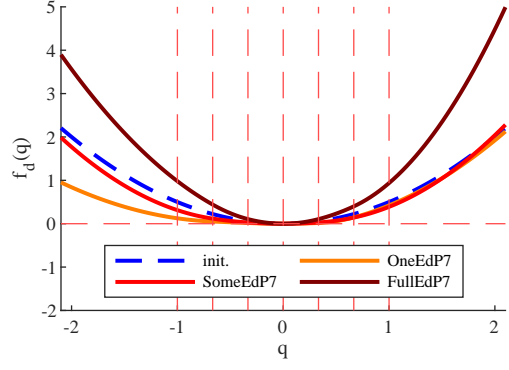}
    \caption{Range reshaping function before (dashed blue) and after (solid) optimization.}
    \label{fig:fdPlot}
  \end{figure}

\begin{remark}\label{remark:discontinuity point}
  Fig.~\ref{fig:5agents_1par} illustrates the change in trajectories our optimization yields. Note that in the bearing-only case (left images), the lack of a fixed scale led to an overall shrinking of the formation to reduce the total path length. This is avoided by including range terms (right images).
%  Through our simulations on bearing-only formation, we observed an undesired side effect of using the total path length as the optimization objective: when the agents start at locations that are nearly opposite (with respect to other agents) to their goals, the optimal solution is to shrink the formation. The reason is bearing-only measurements do not fix the scale of the final formation. With optional range terms, the shrinking behavior can be eliminated.
\end{remark}

% \begin{figure*}[htbp]
% \centering
% \subfloat[\label{fig:5agents_homo_new_path}]{%
%   \includegraphics[width=5.5cm]{figures/homo/NoEr_5agents_homo_statistics_path.eps}%
% }\quad
% \subfloat[\label{fig:5agents_homo_new_pathdiff}]{%
%   \includegraphics[width=5.5cm]{figures/homo/NoEr_5agents_homo_statistics_pathdiff.eps}%
% }\quad
% \subfloat[\label{fig:5agents_homo_new_pathscale}]{%
%   \includegraphics[width=5.5cm]{figures/homo/NoEr_5agents_homo_statistics_sd.eps}%
% }

% \subfloat[\label{fig:5agents_homo_new_pathperc}]{%
%   \includegraphics[width=5cm]{figures/homo/NoEr_5agents_homo_statistics_pathPercent.eps}%
% }\quad
% \subfloat[\label{fig:5agents_homo_new_pathdiffperc}]{%
%   \includegraphics[width=5cm]{figures/homo/NoEr_5agents_homo_statistics_pathdiffPercent.eps}%
% }
% \caption{ Cumulative Density Function and box plot for different performance metrics under 200 test initial conditions. Legend in red is data with $f_b(\cdot)$ before the optimization. Legend in blue, orange, magenta, and green  are the data with $f_b(\cdot)$ after the optimization with 1, 3, 5, and 7 training initial conditions respectively.} 
%   \label{fig:5agents_homo_new}
% \end{figure*}

\subsection{Performance evaluation} \label{test}
We tested the optimized reshaping function in Figs.~\ref{fig:fbPlot},\ref{fig:fdPlot} using 200 randomly selected initial conditions for the agents (called the Test set). Performance was evaluated using five metrics. %Due to the scaling problem in bearing-only control, we evaluated performance in terms of five metrics. 
The first was the path length
$L_{\textrm{path}}(\vct{\alpha},\vct{x}_0)$, defined by \eqref{eqn:objective function} without the terminal condition (see Remark \ref{remark:cost weight}). The second was the difference of the path length from a straight line, 
\begin{equation}\label{metrics:Ldiff}
    L_{\textrm{diff}}(\vct{\alpha},\vct{x}_0)=L_{\textrm{path}}(\vct{\alpha},\vct{x}_0)-\sum_{i\in\cV}\norm{x_i(T)-x_i(0)}.
\end{equation}
% We define scale of a formation as the standard deviation (STD) of the positions of the agents in the formation. From this, the scale change is
% \begin{equation}\label{metrics:std}
%     \delta_{\textrm{scale}} = (STD(\vct{x}(0))-STD(\vct{x}(T))/STD(\vct{x}(0)),
% \end{equation} 
% where $STD(\cdot)$ computes the usual standard deviation for all the elements in an array. %To compare the controllers before ($\vct{\alpha}_\textrm{init}$) and after ($\vct{\alpha}_\textrm{init}$) \sanderss{the before and after thing is exactly the same?} optimization under an initial condition, we define the 
The other two were the corresponding percentage of relative improvement:
\begin{equation}\label{metrics:del_path}
    \delta path =100 \frac{ L_{\textrm{path}}(\vct{\alpha}_\textrm{init}) - L_{\textrm{path}}(\vct{\alpha}_\textrm{opt})}{L_{\textrm{path}}(\vct{\alpha}_\textrm{init})},
\end{equation}   
% and the fifth was the percentage of improvement in the difference of path length and the shortest path (that is, a straight line for from initial to finial position), given by
\begin{equation}\label{metrics:del_diff}
        \delta diff
        = 100 \frac{ L_{\textrm{diff}}(\vct{\alpha}_\textrm{init}) - L_{\textrm{diff}}(\vct{\alpha}_\textrm{opt})}{L_{\textrm{diff}}(\vct{\alpha}_\textrm{init})}.
\end{equation}

We also define scale of a formation as the standard deviation (STD) of the positions of the agents in the formation. 

%We see from the test results that after the optimization, both the path length and difference to straight line metrics shift to lower values, indicating improvement in these metrics. Regarding Remark \ref{remark:discontinuity point}, additional analysis (not included here due to space reasons) indicates that when using only a single initial condition in $\cX_0$, the scale reduces significantly but that this effect is mitigated by including additional intial conditions in the optimization process. \sanderss{Can remove this paragraph because (a) we need the space and (b) it's covered in the sequence of bullet points below. You need to refer to the table to talk about 'test results' anyway; otherwise the reader does not know what you are referring to.}

\begin{table*}[ht]
\caption{Summary of performance. Results for $\delta{path}$ and $\delta{diff}$ are reported as mean values over the trials of the training and test data, with higher numbers indicating better performance. For the test data, the medians for these metrics are also reported in parentheses. In addition, the test data results include the mean for the path length, with lower numbers indicating better performance. The final column is the percentage of trials in the test data for which path length improved.
%Evaluation of Results: a comparison of the training and testing results according to the average of metrics \eqref{metrics:del_path} and \eqref{metrics:del_diff} (number in bracket are the medians of them), and medians of metrics \eqref{eqn:objective function} and \eqref{metrics:Ldiff}. The last three columns are associated with the percentage of number of cases that benefit from the optimization, and how much they get improved. With respect to $\delta{path}$ and $\delta{diff}$, higher value indicates better performance, while $L_\emph{path}$ and $L_\emph{diff}$ are the opposite. $L_\emph{path}$ and $L_\emph{diff}$ have medians in $46.14$ and $7.59$  before the optimization.
}

\label{table:evaluations}
\centering
\begin{tabular}{llllllllll}
  \toprule \multicolumn{1}{l}{\multirow{2}{*}{\hspace{8pt}Case\hspace{8pt}}}
  &\multicolumn{2}{l}{\multirow{2}{*}{\hspace{8pt} Model}}
  &\multicolumn{2}{c}{Training}
  &\multicolumn{5}{c}{Test} \\
  \cmidrule(lr){4-5} \cmidrule(l){6-10}
  &{}
  &{}
  %training
  &$\delta path$ (\%)%  
  &$\delta diff.$ (\%)% 
  %testing
  &$L_\textrm{path,init}$%
  &$L_\textrm{path,opt}$%
  &$\delta path$(\%)%  
  &$\delta diff.$(\%)%
  &$+$\% \\%
  \midrule
  \multirow{8}{2em}{Train~on~5  Test~on~5}
  &\multirow{4}{2em}{NoEd} 
  &1 ICs &\textcolor{blue}{\bf{24.35}}	&\textcolor{blue}{\bf{54.36}} 
   & 45.17 & 40.79	&\textcolor{blue}{\bf{8.72}} (9.46) &0.94 (11.28) &81 \\

  &{} &3 ICs &8.47	&20.61	
  & 45.17	&40.98	&8.6 (8.43)	&12.68 (16.05)	&87 \\

  &{} &5 ICs &7.43	&17.72	
  &45.17	&41.37 &7.82 (7.57)	&11.9	(15.84)	&87.5 \\

  &{} &7 ICs &7.26	&22.04	
  &45.17	&41.38	&7.87	(7.45)	&\textcolor{blue}{\bf{15.7}} (19.76)	&\textcolor{blue}{\bf{90}}\\
  \cline{2-10}
%   \hline
%   \multirow{4}{2em}{Train~on~5  Test~on~5}
  &OneEd &7 ICs &3.75	&23.03	
  &52.39	&50.11	&4.09	(4.32)	&22.58	(24.05)	&90 \\

  &SomeEd &7 ICs &11.02	&53.67	
  &52.79	&47.82	&9.14	(9)	&49.42 (52.83)	&97.5 \\

  &FullEd &7 ICs &\textcolor{blue}{\bf{11.05}}	&\textcolor{blue}{\bf{61.15	}}
  &52.07	&46.13	&\textcolor{blue}{\bf{11.26}} (11.73)	&\textcolor{blue}{\bf{65.64}} (67.78)	&\textcolor{blue}{\bf{99.5}} \\
 
  \hline
  \multirow{4}{2em}{Train~on~5  Test~on~5  Alt.shape}
  &NoEd &7 ICs &7.26	&22.04	
  &45.41	&41.24	&8.8	(8.19)	&22.7 (21.9)	&95.5\\

  &OneEd &7 ICs &3.75	&23.03	
  &49.75	&47.64	&4.02	(4.25)	&21.56	(22.11)	&89 \\

  &SomeEd &7 ICs &11.02	&53.67	
  &50.4	&46.31	&7.91	(8.28)	&39.12 (39.99)	&94 \\

  &FullEd &7 ICs &\textcolor{blue}{\bf{11.05}}	&\textcolor{blue}{\bf{61.15	}}
  &47.98	&43.48	&\textcolor{blue}{\bf{9.16}} (9.52)	&\textcolor{blue}{\bf{56.55}} (59.62)	&\textcolor{blue}{\bf{99}} \\

%   \hline
%   \multirow{4}{2em}{Tr-5 Te-5 S2}
%   &NoEr &7 ICs &7.26	&22.04	
%   &44.88	&40.12	&\textcolor{blue}{\bf{10.16}}	(9.15)	&27.8 (28.53)	&97\\

%   &OneEr &7 ICs &3.75	&23.03	
%   &51.23	&48.61	&4.96	(5.18)	&27.77	(27.13)	&97.5 \\

%   &SomeEr &7 ICs &11.02	&53.67	
%   &52.28	&47.61	&8.78	(8.81)	&43.95 (43.1)	&98 \\

%   &FullEr &7 ICs &\textcolor{blue}{\bf{11.05}}	&\textcolor{blue}{\bf{61.15	}}
%   &49.35	&45.1	&8.46 (8.71)	&\textcolor{blue}{\bf{54.07}} (56.91)	&\textcolor{blue}{\bf{99}} \\

  \hline
  \multirow{4}{2em}{Train~on~5  Test~on~3}
  &NoEd &7 ICs &7.26	&22.04	
  &21.53	&19.43	&8.76	(7.91)	&29.38 (37.79)	&93\\

  &OneEd &7 ICs &3.75	&23.03	
  &28.78	&26.93	&5.97	(6.25)	&40.5 (47.31)	&92.5 \\

  &SomeEd &7 ICs &11.02	&53.67	
  &28.78	&26.75	&6.44	(6.75)	&40.54 (52.07)	&91 \\

  &FullEd &7 ICs &\textcolor{blue}{\bf{11.05}}	&\textcolor{blue}{\bf{61.15}}	
  &28.96	&25.21	&\textcolor{blue}{\bf{12.66}} (14.14)	&\textcolor{blue}{\bf{87.53}} (92.07)	&\textcolor{blue}{\bf{99.5}} \\
  \bottomrule
\end{tabular}
\end{table*}

The results on the Training and Test sets under the different controllers are summarized in Table~\ref{table:evaluations}.
%We summarize the above results on the Test set together with the same metrics for the initial conditions that were optimized over on the Training set in Table~\ref{table:evaluations}, as well as the bearing+distance cases with seven initial condtions. 
Several trends emerge from these results: 
\begin{itemize}
  \item From the Training set, we see that increasing the number of initial conditions \textit{decreases} the improvement. This makes intuitive sense since the algorithm is  looking for the best performance averaged over more conditions.
  \item As expected, there is gap between training and testing performance.  However this gap diminishes as we increase the number of initial conditions considered in the optimization. These results imply that the reshaping function that yields the shortest path depends on where an agent starts but that averaging over a few initial conditions ensures reasonable performance across a wide range of start locations. As seen in the last column, when using seven initial conditions in training, over 90\% of the initial conditions in the test set showed a reduction in path length.
\item Performance improves as additional edges are added to the controller. This likely arises from the additional information available in the extra range measurements.
%By introducing more range edges, the total performance becomes better. The reason could be with a nonlinear controller, more information gives less convolution to move. Without scaling problem, the optimization over bearing+distance formation tend to give better performance with more initial condtions. However, seven initial conditions do achieve similar training and test perforamce to five initial conditions, and they both get to the point where training performance is similar to test performance, indicating the optimized parameter already fits well to the problem.  
  \item The four rows afterwards indicate that the training results can transfer well to a other formations, yielding similar improvements in the trajectories (see Fig.\ref{fig:5agents_shapes} for an example). The last four rows show that the training results transfer to a 3-agent system as well. Interestingly, with full range data and three agents (final row), the results are very close to straight lines, reaching nearly 90\% relative improvement.
    %It is also found that the training result from 5-agent system can be generalized well to 3-agent system. The last four rows show the results where the optimized parameters on 5-agent system test on 3-agent system. The trends on effect of adding more range measurements hold well. It can also be seen that with full range measurements, the trajectories are almost straight lines with $\delta {diff}$ become also 90\%. Moreover, over 90\% of the testing samples are improved.
\end{itemize}

\begin{figure}[htbp]
    \centering
     \subfloat[]{\includegraphics[width=4.0cm]{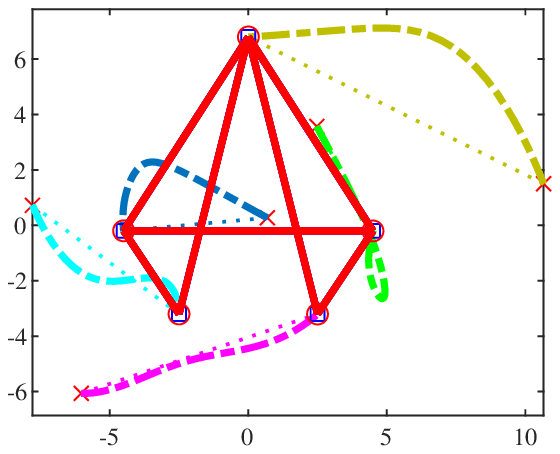}\label{fig:shape1_7par_init}}\quad
    %  \subfloat[]{\includegraphics[width=4.1cm]{figures/homo/FullEr_5agents_7par_init_shape2.eps}\label{fig:shape2_7par_init}}\quad
   %   \subfloat[]{\includegraphics[width=4.2cm]{figures/homo/FullEr_2agents_1par_init.eps}\label{fig:fuller_2agents_1par_init}}
   \subfloat[]{\includegraphics[width=4.0cm]{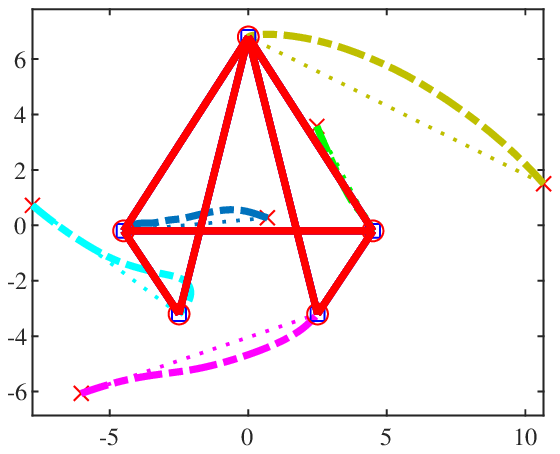}\label{fig:shape1_7par_opt}}\quad
%    \subfloat[]{\includegraphics[width=4.1cm]{figures/homo/FullEr_5agents_7par_opt_shape2.eps}\label{fig:shape2_7par_opt}}\quad
   %   \subfloat[]{\includegraphics[width=4.2cm]{figures/homo/FullEr_2agents_1par_opt.eps}\label{fig:fuller_2agents_1par_opt}}
     \caption{Bearing-based 5-agent system where $E_b=E_d$ with alternate goal configuration. (a) before optimization. (b) after.} 
     \label{fig:5agents_shapes}
   \end{figure}

\subsection{Comparison with Other Controllers}
To further demonstrate the performance of our method, we compare it with other methods under both the bearing-only and bearing+range conditions:
\begin{enumerate}
    \item bearing-only ('NoEd') formation controller: we compare with \cite[eq. (8)]{Zhao_bearingOnly::TAC2016}, a distributed bearing-only formation control law that uses a Lyapunov approach.
    \item bearing+range ($|\cE_d|=|\cE_b|$, 'FullEd') formation controller: we compare with \cite[eq. (5)]{Zhao_bearingBased:ECC2015,Zhao:CSM2019}, a distributed bearing-based formation control law that uses relative position measurements and assumes two independently controlled leaders. Note that we adapt our controller to the leader-follower case by simply fixing the position of the leaders (while still training in the leaderless case).
\end{enumerate}

The results in Table \ref{table:comparison} show that our method yields shorter trajectories, and that the parameters learned from our leaderless optimization are still effective when a leader is considered, despite the fact that \cite{Zhao_bearingBased:ECC2015} actually requires the presence of leaders. %Additionally, our method yields shorter trajectories than the other methods. Furthermore, while \cite{Zhao_bearingBased:ECC2015} requires formation maneuvering using leaders, our controller can be applied in the leader-less setting. 

% Table \ref{table:comparison} presents the results of such comparisons. It demonstrates that our method gives shorter therefore straighter trajectory towards the target formation, noting that our referred controller gives a better performance even without optimization. 
% Also, for the bearing-based comparison, we use our optimized controller from FullEd case, which was not directly optimized for leader-following maneuver. The result might be more obvious if a controller is further trained for this purpose. Furthermore, while \cite{Zhao_bearingBased:ECC2015} requires formation maneuvering using leaders, our controller can be generalized to leaderless maneuver. 

\begin{table}[t] 
  \caption{Comparison: NoEd.Init. and NoEd.Opt. are the bearing-only controllers before and after optimization using 7ICs, while FullEd.Init. and FullEd.Opt. are the ones for bearing+range case. To allow a fair comparison with \protect\cite{Zhao_bearingBased:ECC2015}, methods marked with ${}^\star$ fix the position of two leaders. The reported statistics have the same meaning as in Table \ref{table:evaluations} Test. }
  \label{table:comparison}
  \centering
  \begin{tabular}{p{1.6cm}p{0.75cm}p{0.75cm}p{1.25cm}p{1.25cm}p{0.5cm}}
          \toprule
          &$L_\textrm{path,A}$
          &$L_\textrm{path,B}$
          &$\delta path$(\%) 
          &$\delta diff.$(\%) 
          &$+$\% \\
          \midrule
        %   \begin{tabular}{@{}c}A: NoEd.Init.  \\B: NoEd. Opt.\end{tabular} &45.2 &41.4 &7.9 (7.6) &15.7 (19.8) &90 \\
        %   \midrule
          \begin{tabular}{@{}l@{}}A: Zhao  \\B: NoEd.Init. \end{tabular} &65.9	&45.2	&29	(28.4)	&51.2	(55.4)	&100 \\
          \midrule
          \begin{tabular}{@{}l@{}}A: Zhao \\B: NoEd.Opt.\end{tabular} & 65.9	&41.4	&34.1	(34.1)	&58.2	(63.6)	&99.5\\
          \midrule
        %   \begin{tabular}{@{}l@{}}FullEd.Init.\\vs Opt.\end{tabular} &30.7	&29.5	&3.7	(3.7)	&17.3	(22.6)	&82 \\
        %   \midrule
          \begin{tabular}{@{}l@{}}A${}^\star$: Zhao \\B${}^\star$: FullEd.Init. \end{tabular} &35.5	&30.7	&12.8 (12.9)	&42.9	(7.2)	&94.5 \\
          \midrule
          \begin{tabular}{@{}l@{}}A${}^\star$: Zhao \\B${}^\star$: FullEd.Opt.\end{tabular} &35.5	&29.5	&16.3	(16.9)	&58	(56)	&98.5\\
          \bottomrule
  \end{tabular}
  \end{table}

%averaging the percentage measures \eqref{metrics:del_path} and  \eqref{metrics:del_diff} over the set $\cX_0$ on Table~\ref{table:evaluations}. Noting that in the training, we average sum over $1,2,3,5$ initial conditions and in the test, the number goes up to 200. We can see while the training results in a maximum improvement of 32.7\% in path length and 85.2\% in difference of path length to shortest length, the test results in a maximum improvement of 8.7\% in path length and 13.4\% in difference of path length to shortest path. Among over 80\% being improved by the optimized function, the two metrics can achieve 12.1\% and 23.7\% improvement respectively. There is a clear tradeoff between the training result and the test results. Although using heterogeneous reshaping functions provides better performance in the training set, it doesn't behave any better in the test set comparing to homogeneous reshaping functions.

% \addtolength{\textheight}{-6cm}
\section{Conclusions}
In this paper, we formulated a nonlinear optimization problem on a gradient-descent bearing-based formation controller, which utilizes a reshaping function to embed the relationship between the current and desired bearing. We simulated the algorithms in Matlab to evalulate their performance on a 5-agent network. The optimization consistently led to the same qualitative form of the reshaping function for bearing-only (NoEd) and full range (FullEd) cases, %amely one that was flat when the bearings were far from their desired values and fallowing towards zero as the desired bearings were reached. Interesting, 
though the pattern did not hold for the one-range (OneEd) case. By applying the optimized reshaping function, the bearing-only path length $L_\textrm{path}$ was shortened by around 8\% and the difference to straight lines $L_\textrm{diff}$ was improved by $\sim$16\% in these simulations. In addition, with radomized initial conditions, over 90\% of the trials runs showed improvement relative to the non-optimized case. Our results indicate that using a larger training set in the optimization leads to a small reduction in performance gain but a significant improvement on maintaining the scale of the formation. By introducing more range terms, both the training and test performance improved, up to 11.3\% in $L_\textrm{path}$ and 65.6\% in $L_\textrm{diff}$. More over, almost all the test samples benefited from the optimization despite the small training set. It is also promising to see that the training result on 5-agent network can  be generalized to other formations and networks with different numbers of agents. For future work, we are planning to (1) use the basic form of the optimized reshaping function to reduce the number of control points needed and thus reduce the complexity of the optimization problem, (2) generalize the optimization approach to include different dynamic models for the agents, (3) propose a collision avoidance solution to remove the problem arising when the range between two agents goes to zero.

\bibliographystyle{IEEEtran}

\bibliography{biblio/IEEEfull,biblio/IEEEConfFull,biblio/OtherFull,% Do not insert spaces in this command, otherwise it will not work.
  biblio/zili,%
  biblio/tron,%
  biblio/formationControl,%
  biblio/websites}

\section*{APPENDIX}
\subsection{Proof of Lemma \ref{lemma2}}\label{Appendix1}
 Given \eqref{eq:distance}, \eqref{eq:bearing} and \eqref{eq:edge cost}, we see that $d_{ij}$, $\beta_{ij}$, and $\varphi_{ij}^b$ are invariant to an arbitrary translation $t_{ij}$ (i.e. $d(x_i+t_{ij},x_j+t_{ij})=d(x_i,x_j)$,  $\beta(x_i+t_{ij},x_j+t_{ij})=\beta(x_i,x_j)$,
 $\varphi(x_i+t_{ij},x_j+t_{ij})=\varphi(x_i,x_j)$). Therefore,
 \begin{equation}
    \begin{split}
        \varphi_{ij}^b(\Tilde{x}_i(t),\Tilde{x}_j(t))&=\varphi_{ij}^b(0,\Tilde{x}_j(t)-\Tilde{x}_i(t))\\
        &= \varphi_{ij}^b(0,t(v_j-v_i)).
    \end{split}
\end{equation}
Then, we can move and scale the desired configuration $x_{g,i}$ and $x_{g,j}$ such that $x_{g,i}=(0,0)$, $x_{g,j}=(1,0)$, and $\beta_{g,ij}=(1,0)$. Given $c_{ij}=\beta_{g,ij}^T \beta_{ij}$, we define $v_j-v_i=(c_{ij},\sqrt{1-c_{ij}^2})=\beta_{ij}$. Letting $v_i=(0,0)$, $v_j=(c_{ij},\sqrt{1-c_{ij}^2})$, and fixing agent $i$, the original evaluation on $(\Tilde{x}_i(t),\Tilde{x}_j(t))$ is simplified to evaluate the location of agent $j$ on a radial line starting from the origin. Here, varying $v_j$ is analogous to varying the bearing vector of agent $i$ and $j$ in a unit circle centering at the origin.  
      
It then follows that on the parametric line $(0, tv_j)$, we have
\begin{align}
        v_j&=\frac{\Tilde{x}_j-\Tilde{x}_i}{\norm{\Tilde{x}_j-\Tilde{x}_i}}=\Tilde{\beta}_{ij},\\
    \dot{\Tilde{d}}_{ij}&=\frac{d}{dt}\norm{tv_j-0}=v_j^T \Tilde{\beta}_{ij}.\label{eq:}
\end{align}
Given we are evaluating the terms along a bearing direction, the derivative of $\Tilde{\beta}_{ij}$ on $t$ is a zero vector, and the corresponding $\dot{\Tilde{c}}_{ij}=\beta_{g,ij}^T \dot{\Tilde{\beta}}_{ij}$ is also zero. Then,
\begin{equation}
\begin{split}
    \dot{\Tilde{\varphi}}_{ij}^b &= f_b(\Tilde{c}_{ij}) \dot{\Tilde{d}}_{ij} + \dot{f}_b(\Tilde{c}_{ij}) d_{ij}\dot{\Tilde{c}}_{ij}
    = f_b(\Tilde{c}_{ij}) v_j^T \Tilde{\beta}_{ij}\\ &=f_b(\Tilde{c}_{ij})\Tilde{\beta}_{ij}^T \Tilde{\beta}_{ij}=f_b(\Tilde{c}_{ij}).
    \end{split}
\end{equation}

Given the constraint \eqref{eq:bearing constraints_original1} on $f_b(\cdot)$, we have $\dot{\Tilde{\varphi}}^b_{ij}(t)\geq 0, \text{  with equality if and only if } c_{ij}=1$.

\subsection{Insensitive Weights in the Optimization Objective Terms}\label{Appendix2}
To see this, note that we do not have a regularization term on $\alpha$. As a consequence, we expect that the second term with $\varphi(x(T))$ in the cost \eqref{eqn:objective function} will always be negligible for any reasonable choice of $\omega$. This is because $f_b(\cdot)$, $f'_b(\cdot)$, and hence the control $\vct{u}$, are homogeneous in $\alpha$, which means by scaling $\alpha$ we can maintain the same paths $\vct{x}(t)$, but change the speed; this then implies that the first term in \eqref{eqn:objective function} does not depend on the scale of $\alpha$, but only on its direction. If, by way of contradiction, we had a solution $\alpha_\textrm{opt}$ which was optimal but for which $\varphi(\vct{x}(T)\gg 0$, then we could simply augment the scale of $\alpha$ to make the agents go faster, thus reducing the total cost, giving a contradiction with the fact that $\alpha_\textrm{opt}$ is optimal. 
\subsection{Possible Solution of Discontinuous Sensitivity Funtion}\label{Appendix3}
A (possibly) better approach to handle the discontinuity when the range between two agents goes to zero is to introduce a \textit{bump function} $\phi(d_i)=\prod_{(i,j)\in \cE_b} \phi_e(d_{ij})$ where $d_i=\{d_{ij}\}_{j:(i,j)\in \cE_b}$ indicates the ranges to all agents connected to agent $i$, and
    \begin{equation*}
      \phi_e(d_{ij}) = \begin{cases}
      1-\exp (1+  \frac{1}{ (\frac{ d_{ij} }{ \epsilon })^{2p}  -1}  ), &\text{if $d_{ij}<\epsilon$},\\
      1, &\text{otherwise},
      \end{cases}
      \end{equation*}
      where $\epsilon > 0$ is a given small number and $p$ is an appropriate power. We then modify the path length into $\int^T_0 \phi(d_i)||\dot{x}_i (\vct{x}(t),\vct{\alpha)}|| \de t$. 
     \begin{figure}[htbp]
      \centering
       \subfloat[]{\includegraphics[width=4.05cm]{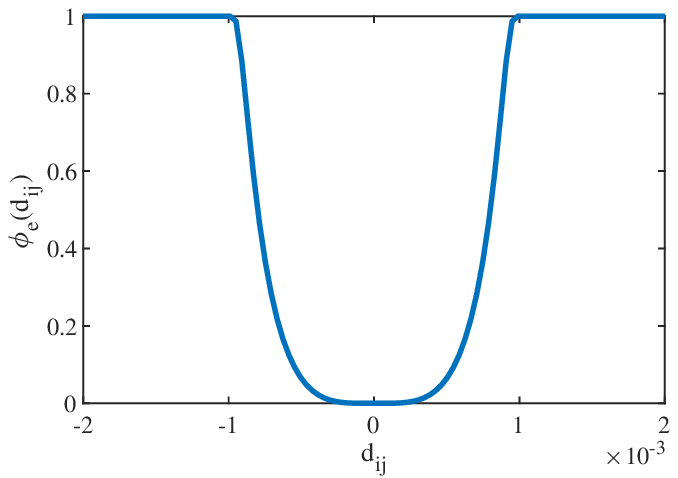}}\quad
       \subfloat[]{\includegraphics[width=4.05cm]{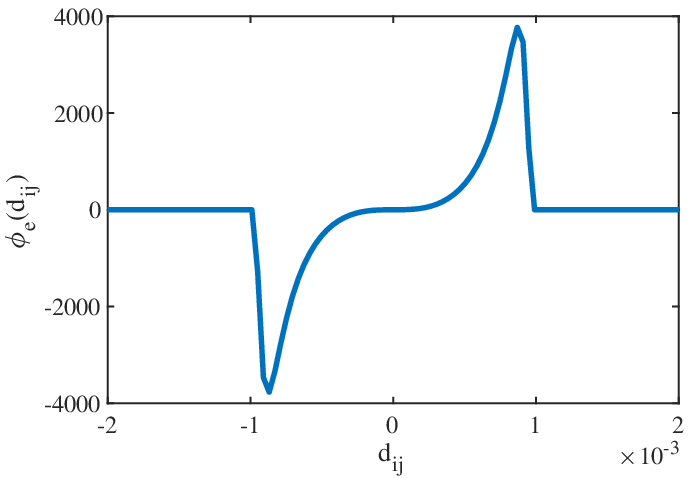}}\quad
       \caption{Bump Function ($p=2$, $\epsilon=10^3$)} 
       \label{fig:bump_function}
     \end{figure}
     
     Since $\phi$ and its derivative are both well behaved (and zero) as two agents move through a common point, including this function should help the sensitivity well-behaved. Intuitively, because the bump function drives the path length term to zero independent of $\alpha$, the integral can be broken into two terms, one before and one after the intersection of the agents. Making this idea rigorous is a topic of ongoing research.
     
%      With this, \eqref{eq:derivative_obj1} becomes
%      \begin{equation}
%       \begin{split}
%           \frac{\partial~}{\partial \alpha}\int_0^T &\phi(d_i) ||\dot{x}_i(\mathbf{x}(t), \mathbf{\alpha})|| dt = \int_0^T \phi(d_i) \frac{x_i}{||\dot{x}_i||}^T \dot{S_i}\\
%           &+\int_0^T \sum_{k\in \{i\cup j:(i,j)\in \cE_b\}} \frac{\partial\phi(d_k)}{\partial x_k} S_k ||\dot{x}_i||~dt.
%       \end{split}
%   \end{equation}

%   Define $t^*$ as the time along a trajectory where $x_i=x_j$. Then
%   \begin{equation}
%     \begin{split}
%         \frac{\partial~}{\partial \alpha}\int_0^T 
%         &\phi(d_i) ||\dot{x}_i(\mathbf{\alpha})|| dt = \frac{\partial~}{\partial \alpha}\int_0^{t^*(\alpha)} \phi(d_i) ||\dot{x}_i( \mathbf{\alpha})|| dt \\
%         &+ \frac{\partial~}{\partial \alpha}\int_{t^*(\alpha)}^T \phi(d_i) ||\dot{x}_i(\mathbf{\alpha})|| dt,
%     \end{split}
% \end{equation}

%     Focusing on the first term, we find
% \begin{equation}
%     \begin{split}
%         \frac{\partial~}{\partial \alpha}\int_0^{t^*(\alpha)} &
%         \phi(d_i) ||\dot{x}_i( \mathbf{\alpha})|| dt =
%         \int_0^{t^*}  \frac{\partial~}{\partial \alpha} \phi(d_i) ||\dot{x}_i (\mathbf{\alpha})|| \\
%         &+ (\phi(d_i) ||\dot{x}_i( \mathbf{\alpha})|| \frac{\partial~}{\partial \alpha}t(\alpha))|_{t=t^*}\\
%         &=\int_0^{t^*}  \frac{\partial~}{\partial \alpha} \phi(d_i) ||\dot{x}_i (\mathbf{\alpha})||,
%     \end{split}
% \end{equation}
% where $t^*$ is a fixed time, the same equation hold for the integral after $t^*$. Therefore, it is valid to restart the sensitivity function when $x_i=x_j$. 

\end{document}

%% file: preamble/common.tex
\input{preamble/commonNoTikz}
\input{preamble/graphicsTikz}
\input{preamble/robotics}

\input{preamble/zili}

%% file: preamble/commonNoTikz.tex
\input{preamble/fixes}
\input{preamble/graphics}
\input{preamble/pagination}
\input{preamble/utilities}
\input{preamble/math}
\input{preamble/operators}

\input{preamble/notation}

\input{preamble/units}

\input{preamble/markupAndCommenting}

%% file: preamble/fixes.tex
\usepackage{fix-cm}
\usepackage{etex}

\usepackage{dblfloatfix}

\usepackage{nag}

%% file: preamble/graphics.tex
\makeatletter
\@ifpackageloaded{xcolor}{}{%
\usepackage[table,x11names,dvipsnames,svgnames]{xcolor}%
}
\makeatother

\usepackage{colortbl}

\usepackage{graphicx}
\usepackage{wrapfig}

\input{preamble/graphicsColors}

%% file: preamble/graphicsColors.tex
\definecolorset{RGB}{lyft}{}{Red,194,39,36;Sunset,202,53,33;Orange,205,68,20;Amber,200,117,42;Yellow,242,169,52;Citron,186,188,44;Lime,112,159,33;Green,56,139,31;Mint,45,118,56;Teal,52,133,135;Cyan,60,132,202;Blue,55,94,248;Indigo,64,13,247;Purple,115,42,248;Pink,176,25,145;Rose,176,32,75}

%% file: preamble/pagination.tex
\usepackage{cite}

\usepackage{microtype}

\usepackage{array}
\usepackage{multirow}
\usepackage{booktabs}
\usepackage{makecell} %introduces \thead and \makecell

\ifcsname labelindent\endcsname
\let\labelindent\relax
\fi
\usepackage[inline]{enumitem}

\usepackage{subfig}

\newenvironment{lenumerate}[2][]
{\begin{enumerate}[label=(#2\arabic*),leftmargin=0.2in,itemindent=0.15in,#1]}
{\end{enumerate}}

\setlist*[enumerate,1]{label={\itshape\arabic*)}}

\makeatletter
\newcommand{\paragraphswithstop}{%
\let\copyparagraph\paragraph%
\renewcommand\paragraph[1]{\copyparagraph{##1.}}%
}
\makeatother

\usepackage[framemethod=tikz]{mdframed}

%% file: preamble/utilities.tex
\usepackage{suffix}

\usepackage{environ}

\makeatletter
\newsavebox{\boxifnotempty}
\newcommand{\displayifnotempty}[3]{\sbox\boxifnotempty{#2}\setbox0=\hbox{\usebox{\boxifnotempty}\unskip}%
\ifdim\wd0=0pt
\else
 #1\usebox{\boxifnotempty}#3%
\fi%
}
\newcommand{\ifempty}[2]{\setbox0=\hbox{#1\unskip}%
\ifdim\wd0=0pt%
 #2%
\fi%
}
\newcommand{\ifnotempty}[2]{\setbox0=\hbox{#1\unskip}%
\ifdim\wd0>0pt%
 #2%
\fi%
}
\makeatother

\usepackage{algpseudocode}
\usepackage{algorithm}

\usepackage{scrlfile}

\makeatletter
\newcommand*\newstoreddef[1]{
  \BeforeClosingMainAux{%
    \immediate\write\@auxout{%
      \string\restoredef{#1}{\csname #1\endcsname}%
    }%
  }%
}
\newcommand*{\restoredef}[2]{% used at the aux file
  \expandafter\gdef\csname stored@#1\endcsname{#2}%
}
\newcommand*{\storeddef}[1]{
  \@ifundefined{stored@#1}{0}{\csname stored@#1\endcsname}%
}
\makeatother

\usepackage{pageslts}
\NewEnviron{tee}{\BODY\typeout{Marker Tee [start] ^^J \BODY ^^JMaker Tee [end]}}

%% file: preamble/operators.tex
\newcommand{\real}[1]{\mathbb{R}^{#1}{}}

\DeclarePairedDelimiter{\norm}{\lVert}{\rVert}

\newcommand{\de}{\mathrm{d}}
\newcommand{\dert}[1][]{\frac{\de #1}{\de t}}

\newcommand{\vct}[1]{\mathbf{#1}}

\DeclareMathOperator{\stack}{stack}

%% file: preamble/notation.tex
\providecommand{\cE}{\mathcal{E}}
\providecommand{\cF}{\mathcal{F}}
\providecommand{\cG}{\mathcal{G}}

\providecommand{\cV}{\mathcal{V}}

\providecommand{\cX}{\mathcal{X}}

%% file: preamble/units.tex
\usepackage{units}

%% file: preamble/markupAndCommenting.tex
\newcommand{\newcolorlabel}[2]{%
  \expandafter\newcommand\csname #1\endcsname[1]{%
    \colorbox{#2}{\color{white}\textsf{\textbf{##1}}}}%
}

\newcommand{\newcommenter}[2]{%
  \expandafter\newcommand\csname #1\endcsname[1]{%
    \fcolorbox{#2}{#2}{\color{white}\textsf{\textbf{#1}}}
    {\color{#2}##1}}%
  \expandafter\newcommand\csname at#1\endcsname{%
    \fcolorbox{#2}{#2}{\color{white}\textsf{\textbf{@#1}}}
    {\color{#2}}}%
  \expandafter\newcommand\csname #1hl\endcsname[2]{%
    \colorbox{#2}{\color{white}\textsf{\textbf{#1}}}\sethlcolor{Azure2}\hl{##2}~%
    \expandafter\ifx\csname commentarrow\endcsname\relax$\leftarrow$\else \commentarrow[#2]\fi~%
    {\color{#2}##1}}%
  \expandafter\newcommand\csname #1st\endcsname[2]{%
    \colorbox{#2}{\color{white}\textsf{\textbf{#1}}}\sout{##2}~%
    \expandafter\ifx\csname commentarrow\endcsname\relax$\leftarrow$\else \commentarrow[#2]\fi~%
    {\color{#2}##1}}%
}
\newcommenter{TODO}{DodgerBlue1}
\newcommenter{rtron}{Green3}

\usepackage{comment}

\usepackage{pdfcomment}

\usepackage{soul}

\usepackage[normalem]{ulem}

\usepackage{csquotes}

%% file: preamble/graphicsTikz.tex
\usepackage{tikz}
\usetikzlibrary{calc}
\usetikzlibrary{matrix}
\usetikzlibrary{chains}
\usetikzlibrary{shapes.geometric}
\usetikzlibrary{arrows.meta}
\usetikzlibrary{decorations.pathreplacing}
\usetikzlibrary{backgrounds}

\tikzset{
  dim above/.style={to path={\pgfextra{
        \pgfinterruptpath
        \draw[>=latex,|->|] let
        \p1=($(\tikztostart)!1.5em!90:(\tikztotarget)$),
        \p2=($(\tikztotarget)!1.5em!-90:(\tikztostart)$)
        in(\p1) -- (\p2) node[pos=.5,sloped,above]{#1};
        \endpgfinterruptpath
      }
    }
  },
  dim double above/.style={to path={\pgfextra{
        \pgfinterruptpath
        \draw[>=latex,|->|] let
        \p1=($(\tikztostart)!3em!90:(\tikztotarget)$),
        \p2=($(\tikztotarget)!3em!-90:(\tikztostart)$)
        in(\p1) -- (\p2) node[pos=.5,sloped,above]{#1};
        \endpgfinterruptpath
      }
    }
  },
  dim below/.style={to path={\pgfextra{
        \pgfinterruptpath
        \draw[>=latex,|->|] let 
        \p1=($(\tikztostart)!-1em!-90:(\tikztotarget)$),
        \p2=($(\tikztotarget)!-1em!90:(\tikztostart)$)
        in (\p1) -- (\p2) node[pos=.5,sloped,below]{#1};
        \endpgfinterruptpath
      }
    }
  },
}

\tikzset{
    right angle quadrant/.code={
        \pgfmathsetmacro\quadranta{{1,1,-1,-1}[#1-1]}     % Arrays for selecting quadrant
        \pgfmathsetmacro\quadrantb{{1,-1,-1,1}[#1-1]}},
    right angle quadrant=1, % Make sure it is set, even if not called explicitly
    right angle length/.code={\def\rightanglelength{#1}},   % Length of symbol
    right angle length=2ex, % Make sure it is set...
    right angle symbol/.style n args={3}{
        insert path={
            let \p0 = ($(#1)!(#3)!(#2)$) in     % Intersection
                let \p1 = ($(\p0)!\quadranta*\rightanglelength!(#3)$), % Point on base line
                \p2 = ($(\p0)!\quadrantb*\rightanglelength!(#2)$) in % Point on perpendicular line
                let \p3 = ($(\p1)+(\p2)-(\p0)$) in  % Corner point of symbol
            (\p1) -- (\p3) -- (\p2)
        }
    }
}

\newcommand{\pgfextractangle}[3]{%
    \pgfmathanglebetweenpoints{\pgfpointanchor{#2}{center}}
                              {\pgfpointanchor{#3}{center}}
    \global\let#1\pgfmathresult  
}

\usetikzlibrary{shapes.arrows}
\newcommand{\commentarrow}[1][Azure4]{\tikz[baseline=-3pt]{\node[shape border uses incircle, fill=#1,rotate=180,single arrow, inner sep=1pt, minimum size=6pt, single arrow head extend=2pt]{};}}

%% file: preamble/robotics.tex
\tikzset{ax/.style={-latex,line width=2pt}}

\tikzset{camera/.style={fill=Sienna1,fill opacity=0.5},%
image plane/.style={draw=RoyalBlue3,line width=2pt}}

%% file: preamble/zili.tex
\usepackage{etoolbox}
\usepackage{tikz}

\newrobustcmd*{\mysquare}[1]{\tikz{\filldraw[draw=#1,fill=none] (0,0)
rectangle (0.2cm,0.2cm);}}

\newrobustcmd*{\mycircle}[1]{\tikz{\filldraw[draw=#1,fill=none] (0,0) circle [radius=0.1cm];}}

\newrobustcmd*{\mytriangle}[1]{\tikz{\filldraw[draw=#1,fill=#1] (0,0) --
(0.2cm,0) -- (0.1cm,0.2cm);}}

\newrobustcmd*{\mycross}[1]{\tikz{\draw[color=#1] (0,0) --(0.2cm,0.2cm);\draw[color=#1] (0.2cm,0) --(0,0.2cm);}}